%% file: polarization.tex
\documentclass[a4paper, 11pt]{article}

\usepackage{amsmath,amssymb,amsthm,epic}
\usepackage{graphicx}
\newtheorem{lemma}{Lemma}
\newtheorem{theorem}{Theorem}
\newtheorem{remark}{Remark}
\newtheorem{corollary}{Corollary}
\newtheorem{observation}{Observation}

\newcommand{\BSC}[1]{{{ \text{BSC}#1}}}
\title{A Class of Transformations that Polarize Symmetric Binary-Input Memoryless Channels}
\author{Satish Babu Korada\quad and\quad Eren \c Sa\c so\u glu}

\begin{document}
\maketitle
\begin{abstract}

A generalization of Ar\i kan's polar code construction using transformations of the form $G^{\otimes n}$ where $G$ is an $\ell \times \ell$ matrix is considered. Necessary and sufficient conditions are given for these transformations to ensure channel polarization. It is shown that a large class of such transformations polarize symmetric binary-input memoryless channels. 
\end{abstract}
\section{Introduction}
Polar codes, introduced by Ar\i kan in \cite{Ari08}, are the first provably capacity achieving codes for arbitrary symmetric binary-input discrete memoryless channels (B-DMC) with low encoding and decoding complexity. Polar code construction is based on the following observation: Let 
\begin{align}\label{eqn:2by2}
G_2=\left[
\begin{array}{cc}
1 & 0 \\
1 & 1
\end{array}
\right].
\end{align}
Consider applying the transform $G_2^{\otimes n}$ (where ``$\phantom{}^{\otimes n}$'' denotes the $n^{th}$ Kronecker power) to a block of $N = 2^n$ bits and transmitting the output through independent copies of a B-DMC $W$ (see Figure \ref{fig:transform}). 
 As $n$ grows large, the channels seen by individual bits (suitably defined in \cite{Ari08}) start \emph{polarizing}: they approach either a noiseless channel or a pure-noise channel, where the fraction of channels becoming noiseless is close to the symmetric mutual information $I(W)$.

It was conjectured in \cite{Ari08} that polarization is a general phemonenon, and is not restricted to the particular transformation $G_2^{\otimes n}$. In this note we give a partial affirmation to this conjecture. In particular, we consider transformations of the form $G^{\otimes n}$ where $G$ is an $\ell\times\ell$ matrix for $\ell \geq 3$ and provide necessary and sufficient conditions for such $G$s to polarize symmetric B-DMCs. 

\begin{figure}[ht]\label{fig:transform}
\input{transform.tex}
\caption{ }
\end{figure}

\section{Preliminaries}

Let $W: \{0,1\} \to \mathcal{Y}$ be a B-DMC. Let $I(W) \in [0,1]$ denote the mutual information between the input and output of $W$ with uniform distribution on the inputs. Also let $Z(W) \in [0,1]$ denote the Bhattacharyya parameter of $W$, i.e., $Z(W) = \sum_{y\in\mathcal{Y}} \sqrt{W(y|0)W(y|1)}$. 

Fix an $\ell\geq 3$ and an invertible $\ell\times\ell$ $\{0,1\}$ matrix $G$. Consider a random $\ell$-vector $U_1^\ell$ that is uniformly distributed over $\{0,1\}^\ell$. Let $X_1^\ell = U_1^\ell G$, where the multiplication is performed over GF(2). Also let $Y_1^\ell$ be the output of $\ell$ uses of $W$ with the input $X_1^\ell$. Observe now that the channel between $U_1^\ell$ and $Y_1^\ell$ is defined by the transition probabilities \begin{align*}
W_\ell (y_1^\ell\mid u_1^\ell) \triangleq \prod_{i=1}^\ell W(y_i\mid x_i) = \prod_{i=1}^\ell W(y_i\mid (u_1^\ell G)_i).
\end{align*}
Define $W^{(i)}: \{0,1\} \to \mathcal{Y}^\ell \times \{0,1\}^{i-1}$ as the channel with input $u_i$, output $(y_1^\ell,u_1^{i-1})$ and transition probabilities
\begin{align*}
W^{(i)}(y_1^\ell,u_1^{i-1}\mid u_i) = \frac{1}{2^{\ell-1}} \sum_{u_{i+1}^\ell} W_\ell(y_1^\ell\mid u_1^\ell),
\end{align*}
and let $Z^{(i)}$ denote its Bhattacharyya parameter, i.e., 
\begin{align*}
Z^{(i)} = \sum_{y_1^\ell,u_1^{i-1}} \sqrt{W^{(i)}(y_1^\ell,u_1^{i-1}\mid 0)W^{(i)}(y_1^\ell,u_1^{i-1}\mid 1)}.
\end{align*}
For $k\geq 1$, let $W^k: \{0,1\}\to\mathcal{Y}^k$ denote the B-DMC with transition probabilities 
\begin{align*}
W^k(y_1^k\mid x) = \prod_{j=1}^k W(y_j\mid x).
\end{align*}
Also let $\tilde{W}^{(i)}: \{0,1\} \to \mathcal{Y}^\ell$ denote the B-DMC with transition probabilities
\begin{align} 
\tilde{W}^{(i)}(y_1^\ell\mid u_i) = \frac{1}{2^{\ell-i}} \sum_{u_{i+1}^\ell} W_\ell(y_1^\ell\mid 0_1^{i-1},u_i^\ell).
\end{align} 
\begin{observation}\label{obs:equivalent}
If $W$ is symmetric, then the channels $W^{(i)}$ and $\tilde{W}^{(i)}$ are equivalent in the sense that for any fixed $u_1^{i-1}$ there exists a permutation $\pi_{u_1^{i-1}}: \mathcal{Y}^\ell \to \mathcal{Y}^\ell$ such that 
\[
W^{(i)}(y_1^\ell, u_1^{i-1} \mid u_i) = \frac{1}{2^{i-1}}\tilde{W}^{(i)}(\pi_{u_1^{i-1}}(y_1^\ell)\mid u_i).
\]
\end{observation}
Finally, let $I^{(i)}$ denote the mutual information between the input and output of channel $W^{(i)}$. Since $G$ is invertible, it is easy to check that 
\begin{align*}
\sum_{i=1}^\ell I^{(i)} = \ell I(W). 
\end{align*}

\section{Polarization}


We will say that $G$ is a \emph{polarizing} matrix if there exists an $i\in\{1,\dotsc,\ell\}$ for which $\tilde{W}^{(i)}$ is equivalent to $W^k$ for some $k\geq 2$, in the sense that
\begin{align}\label{eqn:polarizing}
\tilde{W}^{(i)}(y_1^\ell\mid u_i) = c \prod_{j\in A}W(y_j \mid u_i) 
\end{align}
for some constant $c$ and $A\subseteq \{1,\dots,\ell\}$ with $\vert A \vert = k$. If $W$ is symmetric, then Observation \ref{obs:equivalent} implies the equivalence of $W^{(i)}$ and $W^k$ (which we denote by $W^{(i)} \equiv W^k$) in the sense that 
\begin{align}
W^{(i)}(y_1^\ell,u_1^{i-1} \mid u_i) = \frac{c}{2^{i-1}} \prod_{j\in A}W((\pi_{u_1^{i-1}}(y_1^\ell))_j\mid u_i).
\end{align}
Note that the equivalence $W^{(i)} \equiv W^k$ implies $I^{(i)} = I(W^k)$ and $Z^{(i)} = Z(W^k)$.

It will be shown that channel transformations of the form $G^{\otimes n}$ polarize symmetric channels if and only if $G$ is polarizing. This statement is made precise in the following theorem:

\begin{theorem} \label{thm:main}
Fix a symmetric B-DMC $W$. Let $G^{\otimes n}$ denote the $n^{th}$ Kronecker power of $G$ and consider the transformation $G^{\otimes n}: W \to (W^{(i)}: i=1,\dotsc, \ell^n)$. 
\begin{itemize}
\item[i.] If $G$ is polarizing, then for any $\delta > 0$
\begin{align*}
\lim_{n \to \infty}\frac{\#\left\{i\in\{1,\dotsc,\ell^n\}: I(W^{(i)}) \in (\delta, 1-\delta) \right\}}{\ell^n} = 0.
\end{align*}
\item[ii.] If $G$ is not polarizing, then 
\begin{align*}
I(W^{(i)}) = I(W) \textrm{ for all } n \textrm{ and } i\in\{1,\dotsc,\ell^n\}.
\end{align*}
\end{itemize}
\end{theorem}

Theorem \ref{thm:main} is a direct consequence of Lemmas \ref{lemma:square} and \ref{lemma:I-converges} below.

Note that any invertible $\{0,1\}$ matrix $G$ can be written as a (real) sum $G=P+P'$, where $P$ is a permutation matrix, and $P'$ is a $\{0,1\}$ matrix. This fact can be inferred from Hall's Theorem \cite[Theorem 16.4.]{BoM08}. Therefore, for any such matrix $G$, there exists a column permutation that results in $G_{ii} = 1$ for all $i$. Since the transition probabilities defining $W^{(i)}$ are invariant (up to a permutation of the outputs $y_1^\ell$) under column permutations on $G$, we only consider matrices with 1s on the diagonal. 

The following lemma gives necessary and sufficient conditions for (\ref{eqn:polarizing}) to be satisfied:

\begin{lemma}\label{lemma:square}
For any symmetric B-DMC $W$,
\begin{itemize} 
\item[i.] If $G$ is not upper triangular, then there exists an $i$ for which $W^{(i)} \equiv W^k$ for some $k\geq 2$.
\item[ii.] If $G$ is upper triangular, then $W^{(i)} \equiv W$ for all $1\leq i \leq \ell$.
\end{itemize} 
\end{lemma}
\begin{proof}
Let $G^{(\ell-i)}$ be the $(\ell-i)\times (\ell-i)$ matrix obtained from $G$ by removing its last $i$ rows and columns. Let the number of 1s in the last row of $G$ be $k$. Clearly $W^{(\ell)} \equiv W^k$.  If $k\geq2$ then $G$ is not upper triangular and the first claim of the lemma holds. If $k=1$, then $W^{(\ell)} \equiv W$, and $(x_1,\dotsc,x_{\ell-1})$ is independent of $u_\ell$. One can then write
\begin{align*}
W^{(\ell-i)}(y_1^\ell,u_1^{\ell-i-1}\mid u_{\ell-i}) &= \frac{1}{2^{\ell-1}} \sum_{u_{\ell-i+1}^\ell} W_\ell(y_1^\ell \mid u_1^\ell) \\
& = \frac{1}{2^{\ell-1}} \sum_{u_{\ell-i+1}^{\ell-1},u_\ell} \Pr[Y_1^{\ell-1}= y_1^{\ell-1}\mid U_1^{\ell} = u_1^{\ell}] \Pr[Y_\ell = y_\ell\mid Y_1^{\ell-1} = y_1^{\ell-1},U_1^{\ell} = u_1^{\ell}]\\
&\stackrel{(a)}{ =} \frac{1}{2^{\ell-1}} \sum_{u_{\ell-i+1}^{\ell-1},u_\ell} W_{\ell-1}(y_1^{\ell-1}\mid u_1^{\ell-1}) \Pr[Y_\ell = y_\ell\mid Y_1^{\ell-1} = y_1^{\ell-1},U_1^{\ell} = u_1^{\ell}]\\
& = \frac{1}{2^{\ell-1}} \sum_{u_{\ell-i+1}^{\ell-1}}W_{\ell-1}(y_1^{\ell-1}\mid u_1^{\ell-1}) \sum_{u_\ell} \Pr[Y_\ell = y_\ell\mid Y_1^{\ell-1} = y_1^{\ell-1},U_1^{\ell} = u_1^{\ell}]\\
& = \frac{1}{2^{\ell-1}} \big[W(y_\ell\mid 0) + W(y_\ell\mid 1)\big]
\sum_{u_{\ell-i+1}^{\ell-1}} W_{\ell-1}(y_1^{\ell-1}\mid u_1^{\ell-1})
\end{align*}

where $(a)$ follows from the fact that $G_{lk} = 0,$ for all $k<\ell$. Therefore $y_\ell$ is independent of the inputs to the channels $W^{(\ell-i)}$ for $i=1,\dotsc,\ell-1$. This is equivalent to saying that channels $W^{(1)},\dotsc,W^{(\ell-1)}$ are defined by the matrix $G^{(\ell-1)}$. Applying the same argument to $G^{(\ell-1)}$ and repeating, we see that if $G$ is upper triangular, then we have $W^{(i)} \equiv W$ for all $i$. On the other hand, if $G$ is not upper triangular, then there either exists an $i$ for which $G^{(\ell-i)}$ has at least two 1s in the last row, which in turn implies $W^{(i)} \equiv W^k$ for some $k\geq 2$. \end{proof}

\begin{remark}
The above lemma says that all transformations that are not upper triangular are polarizing. Moreover, upper triangular transformations have no effect on the channel, i.e., each bit sees an independent copy of $W$ after an upper triangular transformation. 
\end{remark}

\begin{corollary}
For any polarizing transformation $G$, there exists an $i\in\{1,\dotsc,\ell\}$ and $k\geq 2$ for which 
\begin{align}
I^{(i)} & = I(W^k) \label{eqn:prod-channel} \\
Z^{(i)} & = Z(W)^k.
\end{align}
\end{corollary}
\begin{proof}
The first claim is trivial. The second claim follows from the fact that the Bhattacharyya parameter of any channel of the form $\prod_j W_j$ is given by $\prod_j Z(W_j)$.
\end{proof}

\section{Convergence}
Consider recursively combining channels $W$ as in \cite{Ari08}, using a polarizing transformation $G$. Following Ar\i kan, associate to this construction a tree process $\{W_n; n\geq 0\}$ with 
\begin{align*}
W_0 & = W \\
W_{n+1} & = W_{n}^{(B_{n+1})},
\end{align*}
where $\{B_n; n\geq 1\}$ is a sequence of i.i.d.\ random variables defined on a probability space $(\Omega, \mathcal{F}, \mu )$, $B_n$ being uniformly distributed over the set $\{1,\dotsc,\ell\}$. Define $\mathcal{F}_0 = \{\emptyset, \Omega\}$ and $\mathcal{F}_n = \sigma (B_1,\dotsc, B_n)$ for $n \geq 1$. Define the processes $\{I_n; n\geq 0\} = \{I(W_n); n\geq 0\}$ and $\{Z_n; n\geq 0\} = \{Z(W_n); n\geq 0\}$. 

\begin{observation}
$\{(I_n, \mathcal{F}_n)\}$ is a bounded martingale and therefore converges a.s.\ and in $\mathcal{L}^1$ to a random variable $I_\infty$.
\end{observation}

\begin{lemma} \label{lemma:I-converges}
If $W$ is symmetric and $G$ is polarizing, then
\begin{align*}
I_\infty =
\begin{cases}
1 & \textrm{w.p. } I(W), \\
0 & \textrm{w.p. } 1 - I(W).
\end{cases}
\end{align*}
\end{lemma}

\begin{proof}
By the convergence in $\mathcal{L}^1$ of $I_n$ we have $\mathbb{E} [|I_{n+1} - I_n | ] \stackrel{n \to \infty}{\longrightarrow}0$. Since $G$ is a polarizing matrix, Lemma \ref{lemma:square} implies
\begin{align*}
I_{n+1} = I(W_n^k) \textrm{ with probability at least } \frac{1}{\ell},
\end{align*}
for some $k\geq 2$. This in turn implies 
\begin{align}
\mathbb{E} [|I_{n+1} - I_n |] \geq \frac{1}{\ell} \mathbb{E}[I(W_n^k) - I(W_n)] \to 0. \label{eqn:I-repetition}
\end{align}
It is shown in the Appendix that for any symmetric B-DMC $W_n$, if $I(W_n) \in (\delta,1-\delta)$ for some $\delta > 0$, then there exists an $\eta (\delta) > 0$ such that $I(W_n^k) - I(W_n) > \eta (\delta)$. We therefore conclude that convergence in (\ref{eqn:I-repetition}) implies $I_\infty \in \{0,1\}$ w.p.\ 1. The claim on the probability distribution of $I_\infty$ follows from the fact that $\{I_n\}$ is a martingale, i.e., $\mathbb{E} [I_\infty] = \mathbb{E} [I_0] = I(W)$.
\end{proof}

\begin{corollary}
If $W$ is symmetric and $G$ is polarizing, then $\{Z_n\}$ converges a.s.\ to a random variable $Z_\infty$ and
\begin{align*}
Z_\infty =
\begin{cases}
0 & \textrm{w.p. } I(W), \\
1 & \textrm{w.p. } 1 - I(W).
\end{cases}
\end{align*}
\end{corollary}
\begin{proof}
The proof follows from the fact that $I_n \to I_\infty$ a.s.\ and the inequalities \cite{Ari08} 
\begin{align*}
I(Q)^2 + Z(Q)^2 \leq 1 \\
I(Q) + Z(Q) \geq 1.
\end{align*}
for any B-DMC $Q$.
\end{proof}

\begin{theorem}\label{thm:l-ary}
Given a symmetric B-DMC $W$, an $\ell\times\ell$ polarizing matrix $G$, and any $\beta < 1/\ell$,
\begin{align*}
\lim_{n \to \infty} \Pr[Z_n \leq 2^{-\ell^{n\beta}}] = I(W).
\end{align*}
\end{theorem}

\begin{proof}[Proof Idea]
For any polarizing matrix it can be shown that $Z_{n+1}\leq \ell Z_n$  with probability 1 and that $Z_{n+1} \leq Z_n^2$  with probability at least $1/\ell$.
The proof then follows by adapting the proof of \cite[Theorem 3]{ArT08}.
\end{proof}

\section{Discussion}
Using Ar\i kan's rule for choosing the information bits, polar codes of blocklength $N=\ell^n$ can be constructed starting with any polarizing $\ell\times\ell$ matrix $G$. The encoding and successive cancellation decoding complexities of such codes are $O(N\log N)$. Using similar arguments, it is easy to show that polar codes of blocklength $N=\prod_{i=1}^n \ell_i$ can be constructed from generator matrices of the form $\otimes_i G_i$, where each $G_i$ is a polarizing matrix of size $\ell_i\times\ell_i$. The encoding and successive cancellation decoding complexities of these codes are also $O(N\log N)$.

\section*{Appendix}

In this section we prove the following:
\begin{lemma}\label{lem:combinedI}
Let $W$ be a symmetric B-DMC and let $W^k$ be defined as above. 
If $I(W) \in (\delta, 1-\delta)$ for some $\delta > 0$, then there exists an $\eta (\delta)>0$ such that $I(W^k) - I(W) > \eta (\delta)$.
\end{lemma}
We will use the following theorem in proving Lemma \ref{lem:combinedI}:
\begin{theorem}[\cite{SSZ05, HuH05}]\label{thm:extremes}
Let $W_1,\dots,W_k$ be $k$ symmetric B-DMCs with capacities $I_1,\dots,I_k$
respectively. Let $W^{[k]}$ denote the channel with transition probabilities
\begin{align*}
W^{[k]}(y_1^k\mid x) = \prod_{i=1}^k W_i(y_i\mid x).
\end{align*}

Also let $W_{BSC}^{[k]}$ denote the channel with transition probabilities
\begin{align*}
W_{\BSC{}}^{[k]}(y_1^k\mid x) =
\prod_{i=1}^k W_{\BSC{(\epsilon_i)}}(y_i\mid x),
\end{align*} 
where $BSC(\epsilon_i)$ denotes the binary symmetric channel with crossover probability  $\epsilon_i \in [0,\frac{1}{2}]$, $\epsilon_i \triangleq h^{-1}(1-I_i)$, where $h$ denotes the binary entropy function. Then, $I(W^{[k]}) \geq I(W_{\BSC{}}^{[k]})$.
\end{theorem}

\begin{remark}
Consider the transmission of a single bit $X$ using
$k$ independent symmetric B-DMCs $W_1,\dots, W_k$ with capacities
$I_1,\dots,I_k$. Theorem \ref{thm:extremes} states
that over the class of all symmetric channels with given mutual informations, 
the mutual information between the input and the output vector is minimized when each of the
individual channels is a BSC. 
\end{remark}

\begin{proof}[Proof of Lemma \ref{lem:combinedI}]
Let
$\epsilon \in [0,\frac12]$ be the crossover probability of a BSC with
capacity $I(W)$, i.e., $\epsilon = h^{-1}(1-I(W))$. 
Note that for $k\geq 2$, 
\begin{align*}
I(W^k) \geq I(W^2).
\end{align*}
By Theorem \ref{thm:extremes}, we have $I(W^2) \geq
I(W_{BSC(\epsilon)}^2)$. 
A simple computation shows that 
\begin{align*}
I(W_{BSC(\epsilon)}^2) = 1+h(2\epsilon\bar{\epsilon})- 2 h(\epsilon).
\end{align*}
We can then write
\begin{align}
I(W^k) - I(W) & \geq I(W_{BSC(\epsilon)}^2) - I(W) \notag \\
& = I(W_{BSC(\epsilon)}^2) - I(W_{BSC(\epsilon)}) \notag  \\
& = h(2\epsilon \bar{\epsilon})-h(\epsilon).
\end{align}
Note that $I(W) \in (\delta, 1-\delta)$ implies $\epsilon \in (\phi (\delta), \frac 12 - \phi(\delta))$ where $\phi(\delta)>0$, which in turn implies $h(2\epsilon \bar{\epsilon})-h(\epsilon) > \eta(\delta)$ for some $\eta(\delta) > 0$.
\end{proof}

\end{document}

%% file: transform.tex
\setlength{\unitlength}{0.5bp}%
\begin{picture}(320,200)(-330,0)
\put(0,0){\includegraphics[scale=0.5]{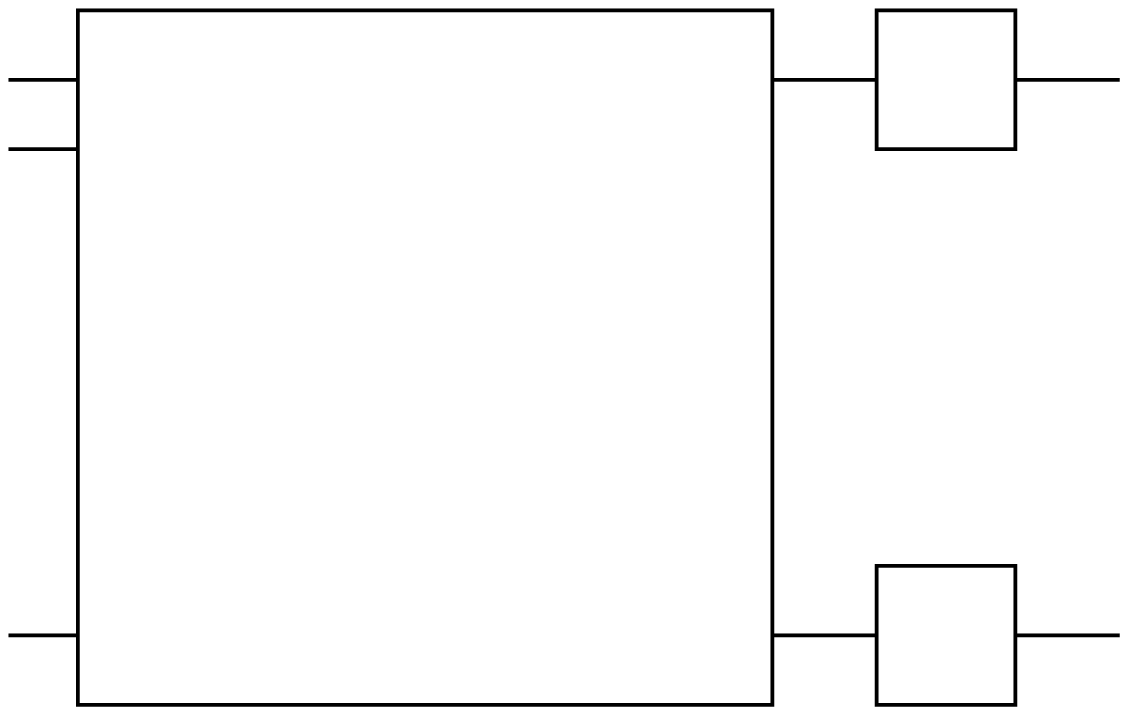}}
\put(262,185){\footnotesize $W$}
\put(265,145){\footnotesize$\cdot$}
\put(265,105){\footnotesize$\cdot$}
\put(265,65){\footnotesize$\cdot$}
\put(262,25){\footnotesize$W$}
\put(115,110){\footnotesize$G^{\otimes n}$}
\put(-40,185){\footnotesize $\textrm{bit}_1$}
\put(-40,165){\footnotesize $\textrm{bit}_2$}
\put(-30,130){\footnotesize $\cdot$}
\put(-30,95){\footnotesize $\cdot$}
\put(-30,60){\footnotesize $\cdot$}
\put(-40,25){\footnotesize $\textrm{bit}_N$}
\end{picture}